\documentclass[runningheads]{llncs}
\usepackage[utf8]{inputenc}
\usepackage{graphicx}
\usepackage{t1enc}
\usepackage{authblk}
\usepackage{amsfonts}
\usepackage{mathtools}
\usepackage{ebproof}
\usepackage{listings}
\usepackage{amssymb}
\usepackage{coqdoc}
\usepackage{amsmath}
\usepackage{url}
\usepackage{hyperref}
\usepackage{float}
\usepackage{cite}
\usepackage{proof}

\newcommand{\ose}{\xrightarrow{e}}
\newcommand{\bos}[4]{|#1, #2, #3| \ose #4}

\newcommand{\Xvar}{\textit{``X''} \: }
\newcommand{\Yvar}{\textit{``Y''} \: }

\numberwithin{equation}{section}

\begin{document}
\title{A Proof Assistant Based \\Formalisation of Core Erlang\thanks{The final authenticated publication is available online at \url{https://doi.org/10.1007/978-3-030-57761-2_7} (Springer link: \url{https://link.springer.com/chapter/10.1007/978-3-030-57761-2_7})}}
%
%
\author{Péter Bereczky\inst{1}\orcidID{0000-0003-3183-0712} \and
Dániel Horpácsi\inst{1}\orcidID{0000-0003-0261-0091} \and
Simon Thompson\inst{1,2}\orcidID{0000-0002-2350-301X}}
\authorrunning{P. Bereczky et al.}
%
\institute{Eötvös Loránd University, HU \\
\email{berpeti@inf.elte.hu} and \email{daniel-h@elte.hu} \and
University of Kent, UK \\
\email{S.J.Thompson@kent.ac.uk}}
\maketitle              

\begin{abstract}	
	Our research is part of a wider project that aims to investigate and reason about the correctness of scheme-based source code transformations of Erlang programs. In order to formally reason about the definition of a programming language and the software built using it, we need a mathematically rigorous description of that language. 
	
	In this paper, we present our proof-assistant-based formalisation of a subset of Erlang, intended to serve as a base for proving refactorings correct. After discussing how we reused concepts from related work, we show the syntax and semantics of our formal description, including the abstractions involved (e.g. closures). We also present essential properties of the formalisation (e.g. determinism) along with their machine-checked proofs. Finally, we prove the correctness of some simple refactoring strategies.

\end{abstract}

\section{Introduction}

There are a number of language processors, development and refactoring tools for mainstream languages, but most of these tools are not theoretically well-founded: they lack a mathematically precise description of what they do to the source code. In particular, refactoring tools are expected to change programs without affecting their behaviour, but in practice, this property is typically verified by regression testing only. Higher assurance can be achieved by making a formal argument (i.e. a proof) about this property, but neither programming languages nor program transformations are easily formalised.

When arguing about behaviour-preservation of program refactoring, we argue about program semantics. To be able to do this in a precise manner, we need a formal, mathematical definition of the programming language semantics in question, which enables formal verification. Unfortunately, most programming languages lack fully formal definitions, which makes it challenging to deal with them in formal ways. Since we are dedicated to improve trustworthiness of Erlang refactorings via formal verification, we put effort in formalising Erlang and its functional core, i.e.\ Core Erlang. Core Erlang is not merely a subset of Erlang; in fact, Erlang (along with other functional languages) translates to Core Erlang as part of the compilation process.

This paper presents work on the Coq formalisation of a big-step semantics for Core Erlang. In general, if formal semantics is not available for a particular language, one can take the language specification and the reference implementation to build a formalisation thereon; in our case, we could rely not only on these artifacts, but also on some previously published semantics definitions. Thus, we reviewed the existing papers on the Core Erlang language and its semantics, distilled, merged and extended them, to obtain a definition that can be properly embedded in Coq. Using that we also proved some basic properties of the semantics as well as proved simple program equivalences.

The main contributions of this paper are:
\begin{enumerate}
\item The definition of a formal semantics for a sequential subset of Erlang (Core Erlang), based partly on existing formalisations.
\item An implementation for this semantics in the Coq proof assistant.
\item Theorems that formalise a number of properties of this formalisation, e.g. determinism, with their machine-checked proofs.
\item Results on program evaluation and equivalence verification using the semantics definition, all formalised in the Coq proof assistant.
\end{enumerate}

\noindent
The rest of the paper is structured as follows. In Section \ref{Sec:work} we review the existing formalisations of Core Erlang and Erlang, and compare them in order to help understand the construction of our formal semantics. In Section \ref{Sec:syntax_semantics} we describe the proposed formal description, including abstractions, syntax, and semantics, while in Section \ref{Sec:application} a number of applications of the semantics are described. Section \ref{Sec:summary} discusses future work, then concludes.

\section{Related work}\label{Sec:work}

Although there have already been a number of attempts to build a full-featured formal definition of the Erlang programming language, the existing definitions show varying language coverage, and only some of them, covering mostly parallel parts of Core Erlang or Erlang, is implemented in a machine-checked proof system.
This alone would provide a solid motivation for the work presented in this paper, but our ultimate goal is to prove refactoring-related theorems in the Coq proof assistant, so our goal is formalise the semantics of Erlang in Coq in a way that enables flawless verification of program equivalence.

In order to reuse existing results, we have reviewed the extensive related work on formalisations of both Erlang~\cite{fredlund2001framework, fredlund2003verification, vidal2014towards, kerl} and Core Erlang~\cite{de2018bounded, harrison2017coerl, lanese2018cauder, lanese2018theory, lanese2019playing, neuhausser2007abstraction, nishida2016reversible}, and tried to incorporate ideas from all of these sources. We have decided to build the formalisation of Core Erlang as a stepping stone toward the definition of the entire Erlang language.

The vast majority of related work presents small-step operational semantics. In particular, one of our former project members has already defined~\cite{kerl} most elements of sequential Erlang in the K specification language, which could be used both for interpretation and for verification. We wish we could reuse this definition in Coq, but for proofs to be carried out in the proof assistant, it is too fine-grained, so we are seeking a big-step operational definition.

As a matter of fact, most papers addressing the formal definition of Erlang focus on the concurrent part (process management and communication primitives) of the language, which is not relevant to our current formalisation goals. Harrison's formalisation of CoErl~\cite{harrison2017coerl} concentrates on the way that communication works, and in particular focusses on how mailboxes are processed. Although the papers dealing with the sequential parts tend to present different approaches to defining the semantics, the elements of the language covered and the syntax used to describe them appears to be very similar in each paper. However, there are slight differences in the level of detail. Some definitions model the language very closely, whilst some do abstract away particular elements; for instance, unlike~\cite{neuhausser2007abstraction},~\cite{lanese2018theory} describes function applications only for function names. 

There is another notable difference in the existing formalisations from the syntax point of view: some define values as a subset of expressions distinguished by defining them in a different syntactic category~\cite{fredlund2001framework, fredlund2003verification, neuhausser2007abstraction, vidal2014towards}, and some define values as ``ground patterns''~\cite{lanese2018cauder, lanese2018theory, lanese2019playing, nishida2016reversible}, i.e. subset of patterns. Both approaches have their advantages and disadvantages, we will discuss this question in more detail in Section \ref{Sec:syntax_semantics}.

We principally used the work by Lanese et al. on defining reversible semantics for Erlang~\cite{lanese2018cauder, lanese2018theory, lanese2019playing, nishida2016reversible}, who define a language ``basically equivalent to a subset of Core Erlang''~\cite{nishida2016reversible}. Although they do not take Core Erlang functions and their closures into consideration (except for top-level functions), which we needed to define from scratch, their work proved to be a good starting point for defining a big-step operational semantics. In addition, we took the Core Erlang Documentation~\cite{carlsson2000core} and the reference compiler for Core Erlang as reference points for understanding the basic abstractions of the language in more detail. When defining function applications, we took some ideas from a paper embedding Core Erlang into Prolog~\cite{de2018bounded}, and when tackling match expressions, the big-step semantics for FMON~\cite{carlier2012first} proved to be useful. Fredlund's fundamental work~\cite{fredlund2001framework} was very influential, but his Erlang formal semantics section discusses parallel parts mainly.

There were some abstractions missing in almost all papers (e.g. the \verb|let| binding with multiple variables, \verb|letrec|, \verb|map| expressions), for which we had to rely on the informal definitions described in~\cite{carlsson2000core} and the reference implementation. Also, in most of the papers, the global environment is modified in every single step of the execution; in contrast, our semantics is less fine-grained as side-effects have been not implemented yet. Unfortunately, the official language specification document was written in 2004, and there were some new features (e.g. the map data type) introduced to Core Erlang since then. These features do not have an informal description either; however, we took the reference implementation and build the formalisation thereon.

\section{Formal semantics of Core Erlang}\label{Sec:syntax_semantics}

After reviewing related work, we present our formal definition of Core Erlang formalised in Coq. Throughout this section, we will frequently quote the Coq definition; in some cases, we use the Coq syntax and quote literally, but in case of the semantic rules, we turned the consecutive implications into inference rule notation for better readability. The entire formalisation is available on Github~\cite{coreerlang}.

\subsection{Syntax}\label{Sec:syntax}

This section gives a brief overview of the syntax in our formalisation.

\begin{figure}[h]
\begin{minipage}{0.45\textwidth}
	\coqdockw{Inductive} \coqdocvar{Literal} : \coqdockw{Type} :=\coqdoceol
	\coqdocnoindent
	\ensuremath{|} \coqdocvar{Atom} (\coqdocvar{s}: \coqdocvar{string}) \coqdoceol
	\coqdocnoindent
	\ensuremath{|} \coqdocvar{Integer} (\coqdocvar{x} : \coqdocvar{Z})\coqdoceol
	\coqdocnoindent
	\ensuremath{|} \coqdocvar{EmptyList}\coqdoceol
	\coqdocnoindent
	\ensuremath{|} \coqdocvar{EmptyTuple}\coqdoceol
	\coqdocnoindent
	\ensuremath{|} \coqdocvar{EmptyMap}.\coqdoceol

	\caption{Syntax of literals}
	\label{fig:lit_syntax}
\end{minipage}
\hfill
\begin{minipage}{0.45\textwidth}
	\coqdockw{Inductive} \coqdocvar{Pattern} : \coqdockw{Type} :=\coqdoceol
	\coqdocnoindent
	\ensuremath{|} \coqdocvar{PVar}     (\coqdocvar{v} : \coqdocvar{Var})\coqdoceol
	\coqdocnoindent
	\ensuremath{|} \coqdocvar{PLiteral} (\coqdocvar{l} : \coqdocvar{Literal})\coqdoceol
	\coqdocnoindent
	\ensuremath{|} \coqdocvar{PList}  (\coqdocvar{hd} \coqdocvar{tl} : \coqdocvar{Pattern})\coqdoceol
	\coqdocnoindent
	\ensuremath{|} \coqdocvar{PTuple} (\coqdocvar{t} : \coqdocvar{list} \coqdocvar{Tuple})\coqdoceol
	\coqdoceol
	
	\caption{Syntax of patterns}
	\label{fig:pat_syntax}
\end{minipage}
\end{figure}

The syntax of literals and patterns (Figures~\ref{fig:lit_syntax} and \ref{fig:pat_syntax}) is based on the papers mentioned in Section~\ref{Sec:work}. The only addition is the \verb|map| construction (\coqdocvar{EmptyMap} literal); float literals are left out, because in our applications, they can be handled as if they were integers. The tuple pattern is represented with Coq's built-in list, which is constructed inductively.

For the definition of the syntax of expressions, we need the following auxiliary type:

\vspace{0.3cm}
\noindent\qquad
\coqdockw{Definition} \coqdocvar{FunctionIdentifier} : \coqdockw{Type} := \coqdocvar{string} \ensuremath{\times} \coqdocvar{nat}.\coqdoceol
\vspace{0.3cm}

With the help of this type alias and the previous definitions, we can describe the syntax of the expressions (Figure~\ref{fig:syntax}). As mentioned in Section~\ref{Sec:work}, our expression syntax is very similar to the existing definitions found in the related work. The main abstractions are based on~\cite{fredlund2001framework, fredlund2003verification, vidal2014towards} and the additional expressions (e.g. \verb|let|, \verb|letrec|, \verb|apply|, \verb|call|) on~\cite{carlsson2000core, lanese2018cauder, lanese2018theory, lanese2019playing, neuhausser2007abstraction, nishida2016reversible}. However, in our formalisation, we included the \verb|map| type, primitive operations and function calls are handled alike, and in addition, the \coqdocvar{ELet} and \coqdocvar{ELetrec} statements handle multiple bindings at the same time.

\begin{figure}[ht]
	\begin{minipage}{\textwidth}
		\coqdockw{Inductive} \coqdocvar{Expression} : \coqdockw{Type} :=\coqdoceol
		\coqdocnoindent
		\ensuremath{|} \coqdocvar{ELiteral} (\coqdocvar{l} : \coqdocvar{Literal})\coqdoceol
		\coqdocnoindent
		\ensuremath{|} \coqdocvar{EVar}     (\coqdocvar{v} : \coqdocvar{Var})\coqdoceol
		\coqdocnoindent
		\ensuremath{|} \coqdocvar{EFunSig}  (\coqdocvar{f} : \coqdocvar{FunctionIdentifier})\coqdoceol
		\coqdocnoindent
		\ensuremath{|} \coqdocvar{EFun}     (\coqdocvar{vl} : \coqdocvar{list} \coqdocvar{Var}) (\coqdocvar{e} : \coqdocvar{Expression})\coqdoceol
		\coqdocnoindent
		\ensuremath{|} \coqdocvar{EList}  (\coqdocvar{hd} \coqdocvar{tl} : \coqdocvar{Expression})\coqdoceol
		\coqdocnoindent
		\ensuremath{|} \coqdocvar{ETuple} (\coqdocvar{l} : \coqdocvar{list} \coqdocvar{Expression}) \coqdoceol
		\coqdocnoindent
		\ensuremath{|} \coqdocvar{ECall}  (\coqdocvar{f}: \coqdocvar{string})     (\coqdocvar{l} : \coqdocvar{list} \coqdocvar{Expression}) \coqdoceol
		\coqdocnoindent
		\ensuremath{|} \coqdocvar{EApply} (\coqdocvar{exp}: \coqdocvar{Expression})     (\coqdocvar{l} : \coqdocvar{list} \coqdocvar{Expression})\coqdoceol
		\coqdocnoindent
		\ensuremath{|} \coqdocvar{ECase}  (\coqdocvar{e} : \coqdocvar{Expression})            (\coqdocvar{l} : \coqdocvar{list} \coqdocvar{Clause})\coqdoceol
		\coqdocnoindent
		\ensuremath{|} \coqdocvar{ELet}   (\coqdocvar{s} : \coqdocvar{list} \coqdocvar{Var})             (\coqdocvar{el} : \coqdocvar{list} \coqdocvar{Expression}) (\coqdocvar{e} : \coqdocvar{Expression})\coqdoceol
		\coqdocnoindent
		\ensuremath{|} \coqdocvar{ELetrec} (\coqdocvar{fnames} : \coqdocvar{list} \coqdocvar{FunctionIdentifier}) (\coqdocvar{fs}\footnote{This is the list of the defined functions (list of variable lists and body expressions)} : \coqdocvar{list} ((\coqdocvar{list} \coqdocvar{Var}) \ensuremath{\times} \coqdocvar{Expression})) (\coqdocvar{e} : \coqdocvar{Expression})\coqdoceol
		\coqdocnoindent
		\ensuremath{|} \coqdocvar{EMap}   (\coqdocvar{kl} \coqdocvar{vl} : \coqdocvar{list} \coqdocvar{Expression})   \coqdoceol
		\coqdocindent{0.50em}
		\coqdockw{with} \coqdocvar{Clause} : \coqdockw{Type} :=\coqdoceol
		\coqdocindent{0.50em}
		\ensuremath{|} \coqdocvar{CCons} (\coqdocvar{p} : \coqdocvar{Pattern})   (\coqdocvar{guard} \coqdocvar{e} : \coqdocvar{Expression}).\coqdoceol
	\end{minipage}
\caption{Syntax of expressions}
\label{fig:syntax}
\end{figure}

\subsubsection{Values}\label{Sec:Values}

In Core Erlang, literals, lists, tuples, maps, and closures can be values, i.e. results of the evaluation of other expressions. As pointed out in Section~\ref{Sec:work}, there are two approaches discussed in the related work: either values are related to patterns~\cite{lanese2018cauder, lanese2018theory, lanese2019playing, nishida2016reversible}
or values are related to expressions~\cite{fredlund2001framework, fredlund2003verification, neuhausser2007abstraction, vidal2014towards}. We have decided to relate values to expressions, because semantically values are derived from expressions and not patterns. Moreover, there are three methods to define the aforementioned relation of values and expressions:

\begin{itemize}
	\item Values are not a distinct syntactic category, so they are defined with an explicit subset relation; 
	\item Values are syntactically distinct and are used in the definition of expressions~\cite{fredlund2001framework,fredlund2003verification,vidal2014towards};
	\item Values are syntactically distinct, but there is no explicit subset relation between values and expressions~\cite{neuhausser2007abstraction}.
\end{itemize}

When values are not defined as a distinct syntactic set (or as a semantic domain), a subset relation has to be defined that tells whether an expression represents a value. In Coq, this subset relation is defined by a judgment on expressions, but this would require a proof every time an expression is handled as a value: the elements of a subset are defined by a pair, i.e. the expression itself and a proof that the expression is a value. While this is a feasible approach, it generates lots of unnecessary  trivial statements to prove in the dynamic semantics: instead of using a list of values, a list of expressions has to be used where proofs must be given about the head and tail being values (see the example in Section \ref{Sec:semantics} for more details about list evaluation). In addition, the main issue with these approaches is that values do not always form a proper subset of patterns or expressions~\cite{carlsson2000core}: when lambda functions and function identifiers (signatures) are considered, values must include closures, which, on the other hand, are not present in the syntax.

For the reasons above, we define values separately from syntax, but unlike~\cite{neuhausser2007abstraction}, we  include function closures in the definition rather than  functions themselves. In fact, we define values as a semantic domain, to which expressions are evaluated (see Figure~\ref{fig:value_syntax}). This distinction of values allows the semantics to be defined in a big-step way with domain changing (from expressions to values). Naturally, this approach causes duplication in the syntax definition (i.e. value syntax is not reused, unlike in~\cite{fredlund2001framework,fredlund2003verification,vidal2014towards}), but it saves a lot when proving theorems about values.

\begin{figure}[h]
	\begin{minipage}{\textwidth}
		\coqdockw{Inductive} \coqdocvar{Value} : \coqdockw{Type} :=\coqdoceol
		\coqdocnoindent
		\ensuremath{|} \coqdocvar{VLiteral} (\coqdocvar{l} : \coqdocvar{Literal})\coqdoceol
		\coqdocnoindent
		\ensuremath{|} \coqdocvar{VClosure}\footnote{A closure represents a function definition together with an environment representing the context in which the function was defined: \coqdocvar{ref} will be the environment or a reference to it, \coqdocvar{vl} will be the function parameter list and \coqdocvar{e} will be the body expression.  \coqdocvar{Environment} is defined in Section~\ref{Sec:semantics} below.} (\coqdocvar{ref} : \coqdocvar{Environment} + \coqdocvar{FunctionIdentifier}) (\coqdocvar{vl} : \coqdocvar{list} \coqdocvar{Var}) (\coqdocvar{e} : \coqdocvar{Expression})\coqdoceol
		\coqdocnoindent
		\ensuremath{|} \coqdocvar{VList} (\coqdocvar{vhd} \coqdocvar{vtl} : \coqdocvar{Value})\coqdoceol
		\coqdocnoindent
		\ensuremath{|} \coqdocvar{VTuple} (\coqdocvar{vl} : \coqdocvar{list} \coqdocvar{Value})\coqdoceol
		\coqdocnoindent
		\ensuremath{|} \coqdocvar{VMap} (\coqdocvar{kl} \coqdocvar{vl} : \coqdocvar{list} \coqdocvar{Value}).\coqdoceol
	\end{minipage}
	\caption{Syntax of values}
	\label{fig:value_syntax}
\end{figure}

In the upcoming sections, we will use the following syntax shortcuts:

\begin{align*}
\textit{tt} := \textit{VLiteral}\ (\textit{Atom}\ \textit{``true''})
\\
\textit{ff} := \textit{VLiteral}\ (\textit{Atom}\ \textit{``false''})
\end{align*}

\subsection{Semantics}\label{Sec:semantics}

We define a big-step operational semantics for the Core Erlang syntax described in the previous section. In order to do so, we need to define environment types to be included in the evaluation configuration. In particular, we define \emph{environments} which hold values of variable symbols and function identifiers, and separately we define \emph{closure environments} to store closure-local context.

\subsubsection{Environment}

The variable environment stores the bindings made with pattern matching in parameter passing as well as in \verb|let|, \verb|letrec|, \verb|case| (and \verb|try|) statements. Note that the bindings may include both variable names and function identifiers, with the latter being associated with function expressions in normal form (closures). In addition, there are top-level functions in the language, and they too are stored in this environment, similarly to those defined with the \verb|letrec| statement.

Top-level, global definitions could be stored in a separate environment in a separate configuration cell, but we decided to handle all bindings in one environment, because this separation would cause a lot of duplication in the semantic rules and in the actual Coq implementation. Therefore, there is one union type to construct a single environment for function identifiers and variables, both local and global. It is worth mentioning that in our case the environment always stores values since Core Erlang evaluation is strict, i.e. first expressions evaluate to some values, then variables can be bound to these values.

We define the environment in the following way:

\vspace{0.3cm}
\noindent
\coqdockw{Definition} \coqdocvar{Environment} : \coqdockw{Type} := \coqdocvar{list} ((\coqdocvar{Var} + \coqdocvar{FunctionIdentifier}) \ensuremath{\times} \coqdocvar{Value}).\coqdoceol
\vspace{0.3cm}

We denote this mapping by $\Gamma$ in what follows, whilst $\varnothing$ is used to denote the empty environment. We also define a number of helper functions to manage environments, which will be used in formal proofs below. For the sake of simplicity, we omit the actual Coq definitions of these operations and rather provide a short summary of their effect.

\begin{itemize}
	\item \coqdocvar{get\_value} $\Gamma$ \coqdocvar{key}: Returns the value associated with \coqdocvar{key} in $\Gamma$.
	\item \coqdocvar{insert\_value} $\Gamma$ \coqdocvar{key} \coqdocvar{value}: Inserts the \emph{(key,value)} pair into $\Gamma$. If this \coqdocvar{key} has already been present, it will replace the original binding with the new one (according to~\cite{carlsson2000core}, section 6). The next three function is implemented with this replacing insertion.
	\item \coqdocvar{add\_bindings} \coqdocvar{bindings} $\Gamma$: Appends to $\Gamma$ the variable-value bindings given in \coqdocvar{bindings}.
	\item \coqdocvar{append\_vars\_to\_env} \coqdocvar{varlist} \coqdocvar{valuelist} $\Gamma$: It is used for \verb|let| statements and adds the bindings (\coqdocvar{varlist} elements to \coqdocvar{valuelist} elements) to $\Gamma$.
	\item \coqdocvar{append\_funs\_to\_env} \coqdocvar{funsiglist} \coqdocvar{param-bodylist} $\Gamma$: Appends to $\Gamma$ function sig-nature-closure pairs. The closures are constructed from \coqdocvar{param-bodylist} which contains parameter lists and body expressions.
\end{itemize}

\subsubsection{Closure Environment}\label{Sec:closures}

In Core Erlang, function expressions evaluate to closures. Closures have to be modeled in the semantics carefully in order to capture the bindings in the context of the closure properly. The following Core Erlang program shows an example where we need to explicitly store a binding context to closures:

\begin{lstlisting}[language=Haskell]
let X = 5 in
  let Y = fun() -> X in
    let X = 10 in
      apply Y()
\end{lstlisting}

The semantics needs to make sure that we apply static binding here: the function \coqdocvar{Y} has to return \coqdocvar{5} rather than \coqdocvar{10}. This requires the \coqdocvar{Y}'s context to be stored along with its body, which is done by coupling them into a \emph{function closure}.

When evaluating a function expression a closure is created. This is a copy of the current environment, an expression (the function body), and a variable list (the parameters of the function).

This information could be encoded with the \coqdocvar{VClosure} constructor in the \coqdocvar{Value} inductive type using the actual environment (see Figure \ref{fig:value_syntax}), however, this cannot be used when the function is recursive. Here is an example:

\begin{lstlisting}[language=Haskell]
letrec 'f1'/0 = fun() -> apply 'f1'/0()
\end{lstlisting}

In Core Erlang, \verb|letrec| allows definition of recursive functions, so the body of the \verb|'f1'/0| must be evaluated in an environment which stores \verb|'f1'/0| mapped to a closure. But this closure contains the environment in which the body expression must be evaluated and that is the same environment mentioned before. So the this is a recursion in embedded closures in the environment. Here is the problem visualized:

\vspace{0.3cm}
\noindent
\{'f1'/0 : \coqdocvar{VClosure} \{'f1'/0 : \coqdocvar{VClosure} \{'f1'/0 : ...\}\} [] (\verb|apply 'f1'/0()|) \}
\vspace{0.3cm}

We do not apply any syntactical changes to the function body, but we solve this issue by introducing the concept of closure environments. The idea is that the name of the function (variable name or function identifier) is mapped to the application environment (this way, it can be used as a reference). It is enough to encode the function's name with the \coqdocvar{VClosure} constructor. This closure environment can only be used together with the use of the environment and items cannot be deleted from it.

\vspace{0.3cm}
\noindent
\coqdockw{Definition} \coqdocvar{Closures} : \coqdockw{Type} := \coqdocvar{list} (\coqdocvar{FunctionIdentifier} \ensuremath{\times} \coqdocvar{Environment}).\coqdoceol
\vspace{0.3cm}

All in all, closures will ensure that the functions will be evaluated in the right environments (a fully formal example is described in Section \ref{Sec:programs}). There are two ways of using their evaluation environment (\emph{ref} attribute of \coqdocvar{Environment + FunctionIdentifier} type):

\begin{itemize}
	\item Either using the concrete environment from the closure value directly if \coqdocvar{ref} is from the type \coqdocvar{Environment};
	\item Or using the reference and the closure environment to get the evaluation environment when the type of \emph{ref} is \coqdocvar{FunctionIdentifier}.
\end{itemize}

In the next sections, we denote this function-environment mapping with $\Delta$, and $\varnothing$ denotes the empty closure environment. Similarly to ordinary environments, closure environments are managed with a number of simple helper functions; like before, we omit the formal definition of these and provide an informative summary instead.

\begin{itemize}
	\item \coqdocvar{get\_env} \coqdocvar{key} $\Delta$ : Returns the environment associated with \coqdocvar{key} in $\Delta$ if \coqdocvar{key} is a \coqdocvar{FunctionIdentifier}. If \coqdocvar{key} is an \coqdocvar{Environment}, the function simply returns it. This function is implemented with the help of the next function.
	\item \coqdocvar{get\_env\_from\_closure} \coqdocvar{key} $\Delta$: Returns the environment associated with \coqdocvar{key}. If the \coqdocvar{key} is not present in the $\Delta$, it returns $\varnothing$.
	\item \coqdocvar{set\_closure} $\Delta$ \coqdocvar{key} $\Gamma$: Adds \coqdocvar{(key, $\Gamma$)} pair to $\Delta$. If \coqdocvar{key} exists in $\Delta$, its value will be overwritten. Used in the next function.
	\item \coqdocvar{append\_funs\_to\_closure} \coqdocvar{fnames} $\Delta$ $\Gamma$: Inserts a  \coqdocvar{($funid_i$, $\Gamma$)} binding into $\Delta$ for every $funid_i$ function identifier in \coqdocvar{fnames}. 
\end{itemize}

\subsubsection{Dynamic Semantics}

The presented semantics, theorems, tests and proofs are available in Coq on the project's Github repository~\cite{coreerlang}.

With the language syntax and the execution environment defined, we are ready to define the big-step semantics for Core Erlang. The operational semantics is denoted by 

\begin{align}
\nonumber
\bos{\Gamma}{\Delta}{\coqdocvar{e}}{\coqdocvar{v}} ::= \coqdocvar{eval\_expr}\ \Gamma\ \Delta\ \coqdocvar{e}\ \coqdocvar{v}
\end{align}

\noindent
where \coqdocvar{eval\_expr} is the semantic relation in Figure \ref{fig:semantics}. This means that \coqdocvar{e} \coqdocvar{Expression} evaluates to \coqdocvar{v} \coqdocvar{Value} in the environment $\Gamma$ and closure environment $\Delta$.

Prior to presenting the rules of the operational semantics, we define a helper for pointwise evaluation of multiple independent expressions: \coqdocvar{eval\_all} states that a list of expressions evaluates to a list of values.

\vspace{0.3cm}
\noindent
\coqdocvar{eval\_all $\Gamma$ $\Delta$ exps vals} := \coqdoceol
\coqdocvar{length} \coqdocvar{exps} = \coqdocvar{length} \coqdocvar{vals} $\Longrightarrow$ \coqdoceol
(\coqdockw{\ensuremath{\forall}} \coqdocvar{exp} : \coqdocvar{Expression}, \coqdockw{\ensuremath{\forall}} \coqdocvar{val} : \coqdocvar{Value},\coqdoceol
\coqdocindent{2.00em}
\coqdocvar{In} (\coqdocvar{exp}, \coqdocvar{val}) (\coqdocvar{combine} \coqdocvar{exps} \coqdocvar{vals}) \ensuremath{\Longrightarrow}\coqdoceol
\coqdocindent{3.00em}
$\bos{\Gamma}{\Delta}{\coqdocvar{exp}}{\coqdocvar{val}}$)\coqdoceol
\vspace{0.3cm}

With the help of this proposition, we will be able to define the semantics of function calls, tuples, and expressions of other kinds in a more readable way. In this definition, we reuse \coqdocvar{length}, \coqdocvar{combine}, \coqdocvar{nth} and \coqdocvar{In} from Coq's built-ins~\cite{coqdocref}.

There is another auxiliary definition which will simplify the main definition: (\coqdocvar{match\_clause} (\coqdocvar{v} : \coqdocvar{Value}) (\coqdocvar{cs} : \coqdocvar{list Clause}) (\coqdocvar{i} : \coqdocvar{nat})) tries to match the \coqdocvar{i}th pattern given in the list of clauses (\coqdocvar{cs}) with the value \coqdocvar{v}. The result is optional; if the \coqdocvar{i}th clause does not match the value, it returns nothing or on successful matching it returns the guard and body expressions with the pattern variable-value bindings from the \coqdocvar{i}th clause.

\begin{figure}[ht!]
	\fontsize{9}{10}
	\coqdockw{Inductive} \coqdocvar{eval\_expr} : \coqdocvar{Environment} \ensuremath{\rightarrow} \coqdocvar{Closures} \ensuremath{\rightarrow} \coqdocvar{Expression} \ensuremath{\rightarrow} \coqdocvar{Value} \ensuremath{\rightarrow} \coqdockw{Prop} :=\coqdoceol
	
	\begin{minipage}{0.46\textwidth}
		\begin{equation}
		\begin{prooftree}
		\infer0{\bos{\Gamma}{\Delta}{\coqdocvar{ELiteral}\ \coqdocvar{l }}{\coqdocvar{VLiteral}\ \coqdocvar{l}}}
		\end{prooftree}
		\label{OS:lit}
		\end{equation}
	\end{minipage}
	\hfill
	\begin{minipage}{0.53\textwidth}
		\begin{equation}
		\begin{prooftree}
		\infer0{\bos{\Gamma}{\Delta}{\coqdocvar{EVar}\ \coqdocvar{s}}{\coqdocvar{get\_value}\  \Gamma\ (\coqdocvar{inl}\ \coqdocvar{s})}}
		\end{prooftree}
		\label{OS:var}
		\end{equation}
	\end{minipage}
	
	\vspace*{0.2cm}
	
	\begin{equation}
	\begin{prooftree}
	\infer0{\bos{\Gamma}{\Delta}{\coqdocvar{EFunSig}\ \coqdocvar{fsig}}{\coqdocvar{get\_value}\ \Gamma\ (\coqdocvar{inr}\ \coqdocvar{fsig})}}
	\end{prooftree}
	\label{OS:funsig}
	\end{equation}
	
	\begin{equation}
	\begin{prooftree}
	\infer0{\bos{\Gamma}{\Delta}{\coqdocvar{EFun}\ \coqdocvar{vl}\ \coqdocvar{e}}{\coqdocvar{VClosure}\ (\coqdocvar{inl}\ \Gamma)\ \coqdocvar{vl}\ \coqdocvar{e}}}
	\end{prooftree}
	\label{OS:fun}
	\end{equation}
	
	\begin{minipage}{0.49\textwidth}
		\begin{equation}
		\begin{prooftree}
		\hypo{eval\_all\ \Gamma\ \Delta\ \coqdocvar{exps}\ \coqdocvar{vals}}
		\infer1{\bos{\Gamma}{\Delta}{\coqdocvar{ETuple}\ \coqdocvar{exps}}{\coqdocvar{VTuple}\ \coqdocvar{vals}}}
		\end{prooftree}
		\label{OS:tuple}
		\end{equation}
	\end{minipage}
	\hfill
	\begin{minipage}{0.49\textwidth}
		\begin{equation}
		\begin{prooftree}
		\hypo{\bos{\Gamma}{\Delta}{\coqdocvar{hd}}{\coqdocvar{hdv}}}
		\hypo{\bos{\Gamma}{\Delta}{\coqdocvar{tl}}{\coqdocvar{tlv}}}
		\infer2{\bos{\Gamma}{\Delta}{\coqdocvar{EList}\ \coqdocvar{hd}\ \coqdocvar{tl}}{\coqdocvar{VList}\ \coqdocvar{hdv}\ \coqdocvar{tlv}}}
		\end{prooftree}
		\label{OS:list}
		\end{equation}
	\end{minipage}
	
	\vspace*{0.2cm}
	
	For the next rule we introduce $no\_previous\_match\ \coqdocvar{i}\ \Delta\ \Gamma\ \coqdocvar{cs}\ \coqdocvar{v} := (\forall \coqdocvar{j} : \coqdocvar{nat},\ \coqdocvar{j} < \coqdocvar{i} \Longrightarrow (\coqdockw{\ensuremath{\forall}}\ (\coqdocvar{gg}, \coqdocvar{ee} : \coqdocvar{Expression}),\ (\coqdocvar{bb} : \coqdocvar{list}\ (\coqdocvar{Var} \times \coqdocvar{Value}),\ \coqdocvar{match\_clause}\ \coqdocvar{v}\ \coqdocvar{cs}\ \coqdocvar{j} = \coqdocvar{Some}\ (\coqdocvar{gg},\ \coqdocvar{ee},\ \coqdocvar{bb}) \ensuremath{\Longrightarrow} (\bos{\coqdocvar{add\_bindings}\ \coqdocvar{bb}\ \Gamma}{\Delta}{\coqdocvar{gg }}{\coqdocvar{ff} })))$.
	
	\begin{equation}
	\begin{prooftree}
	\hypo{\coqdocvar{match\_clause}\ \coqdocvar{v}\ \coqdocvar{cs}\ \coqdocvar{i} = \coqdocvar{Some}\ (\coqdocvar{guard}, \coqdocvar{exp}, \coqdocvar{bindings})}
	\infer[no rule]1{\bos{\coqdocvar{add\_bindings}\ \coqdocvar{bindings}\ \Gamma}{\Delta}{\coqdocvar{guard}}{\coqdocvar{tt}}}
	\infer[no rule]1{\bos{\coqdocvar{add\_bindings}\ \coqdocvar{bindings}\ \Gamma}{\Delta}{\coqdocvar{exp}}{\coqdocvar{v'}}}
	\infer[no rule]1{\bos{\Gamma}{\Delta}{\coqdocvar{e}}{\coqdocvar{v}}}
	\infer[no rule]1{no\_previous\_match\ \coqdocvar{i}\ \Delta\ \Gamma\ \coqdocvar{cs}\ \coqdocvar{v}}
	\infer1{\bos{\Gamma}{\Delta}{\coqdocvar{ECase}\ \coqdocvar{e}\ \coqdocvar{cs}}{\coqdocvar{v'}}}
	\end{prooftree}
	\label{OS:case}
	\end{equation}
	
	\vspace*{0.2cm}
	
	\begin{equation}
	\begin{prooftree}
	\hypo{eval\_all\ \Gamma\ \Delta\ \coqdocvar{params}\ \coqdocvar{vals}}
	\hypo{\coqdocvar{eval}\ \coqdocvar{fname}\ \coqdocvar{vals} = \coqdocvar{v}}
	\infer2{\bos{\Gamma}{\Delta}{\coqdocvar{ECall}\ \coqdocvar{fname}\ \coqdocvar{params}}{\coqdocvar{v}}}
	\end{prooftree}
	\label{OS:call}
	\end{equation}
	
	\vspace*{0.2cm}
	
	\begin{equation}
	\begin{prooftree}
	\hypo{eval\_all\ \Gamma\ \Delta\ \coqdocvar{params}\ \coqdocvar{vals}}
	\hypo{\bos{\Gamma}{\Delta}{\coqdocvar{exp}}{\coqdocvar{VClosure}\ \coqdocvar{ref}\ \coqdocvar{var\_list}\ \coqdocvar{body}}}
	\infer[no rule]2{\bos{\coqdocvar{append\_vars\_to\_env}\ \coqdocvar{var\_list}\ \coqdocvar{vals}\ (\coqdocvar{get\_env}\ \coqdocvar{ref}\ \Delta)}{\Delta}{body}{v}}
	\infer1{\bos{\Gamma}{\Delta}{\coqdocvar{EApply}\ \coqdocvar{exp}\ \coqdocvar{params}}{\coqdocvar{v}}}
	\end{prooftree}
	\label{OS:apply}
	\end{equation}
	
	\vspace*{0.2cm}
	
	\begin{equation}
	\begin{prooftree}
	\hypo{eval\_all\ \Gamma\ \Delta\ \coqdocvar{exps}\ \coqdocvar{vals}}
	\hypo{\bos{\coqdocvar{append\_vars\_to\_env}\ \coqdocvar{vars}\ \coqdocvar{vals}\ \Gamma}{\Delta}{e}{\coqdocvar{v}}} 
	\infer2{\bos{\Gamma}{\Delta}{\coqdocvar{ELet}\ \coqdocvar{vars}\ \coqdocvar{exps}\ \coqdocvar{e}}{\coqdocvar{v}}}
	\end{prooftree}
	\label{OS:let}
	\end{equation}
	
	\vspace*{0.2cm}
	
	For the following rule we introduce $\Gamma' ::= \coqdocvar{append\_funs\_to\_env}\ \coqdocvar{fnames}\ \coqdocvar{funs}\ \Gamma$
	
	\vspace*{0.2cm}
	
	\begin{equation}
	\begin{prooftree}
	\hypo{\coqdocvar{length}\ \coqdocvar{funs} = \coqdocvar{length}\ \coqdocvar{fnames}}
	\infer[no rule]1{\bos{\Gamma'}{\coqdocvar{append\_funs\_to\_closure}\ \coqdocvar{fnames}\ \Delta\ \Gamma'}{\coqdocvar{e}}{\coqdocvar{v}}} 
	\infer1{\bos{\Gamma}{\Delta}{\coqdocvar{ELetrec}\ \coqdocvar{fnames}\ \coqdocvar{funs}\ \coqdocvar{e}}{\coqdocvar{v}}}
	\end{prooftree}
	\label{OS:letrec}
	\end{equation}
	
	\vspace*{0.2cm}
	
	\begin{equation}
	\begin{prooftree}
	\hypo{eval\_all\ \Gamma\ \Delta\ \coqdocvar{kl}\ \coqdocvar{kvals}}
	\hypo{eval\_all\ \Gamma\ \Delta\ \coqdocvar{vl}\ \coqdocvar{vvals}}
	\hypo{\coqdocvar{length}\ \coqdocvar{kl} = \coqdocvar{length}\ \coqdocvar{vl}}
	\infer3{\bos{\Gamma}{\Delta}{\coqdocvar{EMap}\ \coqdocvar{kl}\ \coqdocvar{vl}}{\coqdocvar{VMap}\ \coqdocvar{kvals}\ \coqdocvar{vvals}}}
	\end{prooftree}
	\label{OS:map}
	\end{equation}
	
	\caption{The big-step operational semantics of Core Erlang}
	\label{fig:semantics}
\end{figure}

The formal definition of the proposed operational semantics for Core Erlang is presented in Figure~\ref{fig:semantics}. We remind the reader that the figure presents the actual Coq definition, but the inductive cases are formatted as inference rules. We also note that this big-step definition is partly based on the small-step definition discussed in~\cite{lanese2018theory, lanese2019playing, nishida2016reversible} and in some aspects on the big-step semantics in~\cite{carlier2012first,de2018bounded}. In addition, for most of the language elements defined an informal definition is available in~\cite{carlsson2000core}. In the next paragraphs, we provide short explanations of the less trivial rules.

\begin{itemize}
	\item Rule \ref{OS:case}: At first, the case expression \coqdocvar{e} must be evaluated to some \coqdocvar{v} value. Then this \coqdocvar{v} must match to the pattern (\coqdocvar{match\_clause} function) of the specified \coqdocvar{i}th clause. This match provides the guard, the body expressions of the clause and also the pattern variable binding list. The guard must be evaluated to \coqdocvar{tt} in the extended environment with the result of the pattern matching (the binding list mentioned before). The \coqdocvar{no\_previous\_match} states, that for every clause before the \coqdocvar{i}th one the pattern matching cannot succeed or the guard expression evaluates in the extended environment to \coqdocvar{ff}. Thereafter the evaluation of the body expression can continue in this environment.
	
	\item Rule \ref{OS:call}: At first, the parameters must be evaluated to values. Then these values are passed to the auxiliary \coqdocvar{eval} function which simulates the behaviour of inter-module function calls (e.g. the addition inter-module call is represented in Coq with the addition of numbers). This results in a value which will be the result of the \coqdocvar{ECall} evaluation.
	
	\item Rule \ref{OS:apply}: This rule works in similar way to the one described in~\cite{de2018bounded} with the addition of closures. To use this rule, first \coqdocvar{exp} has to be evaluated to a closure. Moreover, every parameter must be evaluated to a value. Finally, the closure's body expression evaluates to the result in an extended environment which is constructed from the parameter variable-value bindings and the evaluation environment of the closure. This environment can be acquired from the closure environment indirectly or it is present in the closure value itself (Section \ref{Sec:closures}).
	\item Rule \ref{OS:let}: At first, every expression given must be evaluated to a value. Then the body of the \verb|let| expression must be evaluated in the original environment extended with the variable-value bindings.
	\item Rule \ref{OS:letrec}: From the functions described (a list of variable list and body expressions), closures will be created and appended to the environment and closure environment associated with the given function identifiers (\coqdocvar{fnames}). In these modified contexts the evaluation continues.
	\item Rule \ref{OS:map}: Introduces the evaluation for maps. This rule states that every key in the map's key list and value list must be evaluated to values resulting in two lists of values (for the map keys and their associated values) from which the value map is constructed. 
	In the future, this evaluation must be modified, because the normal form of maps cannot contain duplicate keys, moreover it is ordered based on these keys, according to the reference implementation.
\end{itemize}

After discussing these rules, we show an example where the approach in which values are defined as a subset of expressions is more difficult to work with.
Let us consider a unary operator (\verb|val|) on expressions which marks the values of the expressions. With the help of this operator, the type of values can be defined: 
\begin{align}
\nonumber
\coqdocvar{Value} ::= \{e : \coqdocvar{Expression}\ |\ e\ \coqdocvar{val}\}.
\end{align}
\noindent
Let us consider the key ways in which this  would modify our semantics.

\begin{itemize}
	\item \coqdocvar{Environment} \ensuremath{\rightarrow} \coqdocvar{Closures} \ensuremath{\rightarrow} \coqdocvar{Expression} \ensuremath{\rightarrow} \coqdocvar{Expression} \ensuremath{\rightarrow} \coqdockw{Prop} would be the type of \coqdocvar{eval\_expr}. This way an additional proposition is needed which states that values are expressions in normal form, \emph{i.e.} they cannot be used on the left side of the rewriting rules.
	\item The expressions which are in normal form could not be rewritten.
	\item Function definitions have to be handled as values
	\item Because of the strictness of Core Erlang, the derivation rules change, additional checks are needed in the preconditions, e.g. in the rule \ref{OS:list}:
		\begin{equation}
		\nonumber
		\begin{prooftree}
		\hypo{\coqdocvar{tlv}\ \coqdocvar{val}}
		\infer[no rule]1{\coqdocvar{hdv}\ \coqdocvar{val}}
		\hypo{\bos{\Gamma}{\Delta}{\coqdocvar{hd}}{\coqdocvar{hdv}} \lor \coqdocvar{hd} = \coqdocvar{hdv}}
		\infer[no rule]1{\bos{\Gamma}{\Delta}{\coqdocvar{tl}}{\coqdocvar{tlv}} \lor \coqdocvar{tl} = \coqdocvar{tlv}}
		\infer2{\bos{\Gamma}{\Delta}{\coqdocvar{EList}\ \coqdocvar{hd}\ \coqdocvar{tl}}{\coqdocvar{VList}\ \coqdocvar{hdv}}\ \coqdocvar{tlv}}
		\end{prooftree}
		\label{OS:list2}
		\end{equation}
\end{itemize}

This approach has the same expressive power as the presented one, but it has more preconditions to prove while using it. For reason, argue that our formalisation is easier to use.

\subsubsection{Proofs of properties of the semantics}

We have also managed to formalise and prove theorems about the attributes of the operations, auxiliary functions and the semantics. Here we present two of these together with proof sketches.

\begin{theorem}[Determinism]\\
	\coqdockw{\ensuremath{\forall}} $(\Gamma : \coqdocvar{Environment}), (\Delta : \coqdocvar{Closures})$, (\coqdocvar{e} : \coqdocvar{Expression}), (\coqdocvar{$v_1$} : \coqdocvar{Value}), \\ $\bos{\Gamma}{\Delta}{\coqdocvar{e}}{v_1}$ \ensuremath{\Longrightarrow} (\coqdockw{\ensuremath{\forall}}\coqdocvar{$v_2$} : \coqdocvar{Value}, $\bos{\Gamma}{\Delta}{\coqdocvar{e}}{v_2}$ \ensuremath{\Longrightarrow} \coqdocvar{$v_1$} = \coqdocvar{$v_2$}).
\end{theorem}
\begin{proof}
	Induction by the construction of the semantics.
	\begin{itemize}
		\item \ref{OS:lit}, \ref{OS:var}, \ref{OS:funsig} and \ref{OS:fun} are trivial: e.g. a value literal can only be derivated from its expression counterpart.
		\item \ref{OS:tuple} and \ref{OS:map} are similar, \ref{OS:map} is basically a double tuple. According to the induction hypothesis each element in the expression tuple can be evaluated to a single value, so the tuple itself evaluates to the tuple which contains these values. The proof for maps is similar.
		\item \ref{OS:list} The head and the tail expression of the list can be evaluated to a single head and tail value according to the induction hypotheses. So the list constructed from the head and tail expressions can only be evaluated to the value list constructed from the head and tail values.
		\item \ref{OS:case} The induction hypothesis states that the base and the clause body and guard expressions evaluate deterministically. The clause selector functions are also deterministic, so there is only one possible way to select a body expression to evaluate.
		\item The other cases are similar to those presented above.
	\end{itemize}
\qed
\end{proof}

\begin{theorem}[Commutativity]
\label{Prf:comm}
\coqdockw{\ensuremath{\forall}} (\coqdocvar{v}, \coqdocvar{v'} : \coqdocvar{Value}),\coqdoceol
\coqdocindent{1.00em}
$\coqdocvar{eval}\ \textit{``plus'' } [\coqdocvar{v} ; \coqdocvar{v'} \: ] =
\coqdocvar{eval}\ \textit{``plus'' } [\coqdocvar{v'} ; \coqdocvar{v}\: ].$

\end{theorem}
\begin{proof}
	First we separate cases based on the all possible construction of values (5 constructors, \coqdocvar{v} and \coqdocvar{v'} values, that is 25 cases). In every case where either of the values is not an integer literal, the \verb|eval| function results in the same error value on both side of the equality.
	
	One case is remaining, when both \coqdocvar{v} and \coqdocvar{v'} are integer literals. In this case the definition of \emph{eval} is the addition of these numbers, and the commutativity of this addition has already been proven in the Coq standard library~\cite{coqdocref}.
	\qed
\end{proof}

\section{Application and testing of the semantics}\label{Sec:application}

In the previous section we have defined a big-step operational semantics for the sequential part of the Core Erlang language, which we also formalised in the Coq proof assistant.

In this section we present some use cases. First, we elaborate on the verification of the semantics definition by testing it against the reference implementation of the language, then we show some examples on how we used the formalisation for deriving program behaviour and for proving program equivalence. 

\subsection{Testing of the semantics}

Due to a lack of an up-to-date language specification, we validated the correctness of our semantics definition by comparing it to the behaviour of the code emitted by the official compiler.

To test our formal semantics, we used equivalence partitioning.
We have written tests both in Coq (version 8.11.0) and in Core Erlang (OTP version 22.0) for every type of expression defined in our formalisation, these were the first partitions. Moreover, there have also been  special complex expressions that have needed separate test cases (e.g. using bound variables in \verb|let| expressions, application of recursive functions, etc.), with these we could divide the bigger partitions into smaller ones.

\subsection{Formal program evaluation}\label{Sec:programs}

Now let us demonstrate how Core Erlang programs are evaluated in the formal semantics. For the sake of readability, we use concrete Core Erlang syntax in the proofs, and trivial statements are omitted from the proof tree.

The first example shows how to evaluate a simple expression with binding:

\begin{equation}
\nonumber
\begin{prooftree}
\infer0{\{X : 5\}(X) = 5}
\infer1[\ref{OS:var}]{\bos{\{X : 5\}}{\varnothing}{X}{5}}
\infer1[\ref{OS:let}]{\bos{\varnothing}{\varnothing}{let\ X = 5\ in\ X}{5}}
\end{prooftree}
\end{equation}

The second example is intended to demonstrate the purpose of the closure values. Here at the application of \ref{OS:apply} it is shown that the body of the application is evaluated in the environment given by the closure.
\begin{equation}
\nonumber
\fontsize{9}{12}
\begin{prooftree}
\infer0{\{X : 42\}(X) = 42}
\infer1[\ref{OS:var}]{\bos{\{ X : 42 \}}{\varnothing}{X}{42}}
\infer1[\ref{OS:apply}]{\bos{\{X : 5, Y : \coqdocvar{VClosure}\ (\coqdocvar{inl}\ \{ X : 42 \})\ []\ X \}}{\varnothing}{apply\ Y()}{42}}
\infer1[\ref{OS:let}]{\bos{\{X : 42, Y : \coqdocvar{VClosure}\ (\coqdocvar{inl}\ \{ X : 42 \})\ []\ X \}}{\varnothing}{let\ X\ = 5\ in\ apply\ Y()}{42}}
\infer1[\ref{OS:let}]{\bos{\{X : 42\}}{\varnothing}{let\ Y = fun() \rightarrow X\ in\ let\ X = 5\ in\ apply\ Y()}{42}}
\infer1[\ref{OS:let}]{\bos{\varnothing}{\varnothing}{let\ X = 42\ in\ let\ Y = fun() \rightarrow X\ in\ let\ X = 5\ in\ apply\ Y()}{42}}
\end{prooftree}
\end{equation}

The third example cannot be evaluated in our formalisation, because of infinite recursion. For readability $\Gamma := \{'x'/0 : \coqdocvar{VClosure}\ (\coqdocvar{inr}\ 'x'/0)\ []\ apply\ 'x'/0() \}$ (the environment after the binding is added) is introduced. This example also presents, that to evaluate recursive functions, the evaluation environment can be gotten from the closure environment.

\begin{equation}
\nonumber
\fontsize{9}{12}
\begin{prooftree}
\hypo{...}
\infer1[\ref{OS:apply}]{\bos{\Gamma}{\{'x'/0 : \Gamma\}}{apply\ 'x'/0()}{??}}
\infer1[\ref{OS:apply}]{\bos{\Gamma}{\{'x'/0 : \Gamma\}}{apply\ 'x'/0()}{??}}
\infer1[\ref{OS:letrec}]{\bos{\varnothing}{\varnothing}{letrec\ 'x'/0 = fun() \rightarrow apply\ 'x'/0()\ in\ apply\ 'x'/0()}{??}}
\end{prooftree}
\end{equation}

\subsection{Program equivalence proofs}\label{Sec:equivalence}

Last but not least, let us present some program equivalence proofs demonstrating the usability of this semantics definition implemented in Coq. This is a significant result of the paper since our ultimate goal with the formalisation is to prove refactorings correct.

For the simplicity, we use \verb|+| to refer to the \coqdocvar{append\_vars\_to\_env} function and $e_1 + e_2$ will denote the $\coqdocvar{ECall}\ \textit{``plus'' } [e_1, e_2]$ expression in the following sections.

First, we present a rather simple example of program equivalence.

\begin{example}[Swapping variable values]\label{ex:swap}

\lstset{columns=fullflexible}

\begin{lstlisting}[language=Haskell]
  let X = 5 in let Y = 6 in X + Y
\end{lstlisting}
is equivalent to
\begin{lstlisting}[language=Haskell]
  let X = 6 in let Y = 5 in X + Y
\end{lstlisting}
	
\end{example}
\begin{proof}
	The formal description of the example looks like the following (using abstract syntax for this one step):

\begin{align*}
&\forall t : Value,
\\
&\bos{\varnothing}{\varnothing}{\coqdocvar{ELet } [\Xvar]\ [\coqdocvar{ELiteral } (\coqdocvar{Integer } 5)] (\coqdocvar{ELet } [\Yvar]\ [\coqdocvar{ELiteral } (\coqdocvar{Integer } 6)] 
\\
&\qquad(\coqdocvar{ECall } \textit{``plus'' } [\coqdocvar{EVar } \Xvar ; \coqdocvar{EVar } \Yvar]))}{\coqdocvar{t}}
\Longleftrightarrow
\\
&
\bos{\varnothing}{\varnothing}{\coqdocvar{ELet } [\Xvar]\ [\coqdocvar{ELiteral } (\coqdocvar{Integer } 6)] (\coqdocvar{ELet } [\Yvar]\ [\coqdocvar{ELiteral } (\coqdocvar{Integer } 5)]
\\
&\qquad
(\coqdocvar{ECall } \textit{``plus'' } [\coqdocvar{EVar } \Xvar ; \coqdocvar{EVar } \Yvar]))}{\coqdocvar{t}}
\end{align*}
	
	Both directions of this equivalence are proven exactly the same way, so only the $\Longrightarrow$ direction is presented here. This way, the hypothesis is the left side of the equivalence.
	
	First, this hypothesis should be decomposed. From the two \verb|let| statements, it is known that the 5 and 6 expression literals can be evaluated only to their value counterparts (because of the determinism and the rule \ref{OS:let}). These ones will be associated with X and Y in the evaluation environment for the addition operator (\coqdocvar{ECall} \textit{``plus''}). When this statement is evaluated (rule \ref{OS:call}), then it yields the following hypothesis:

\[
	t = \textit{eval } \textit{``plus'' } [\textit{VLiteral } (\textit{Integer } 5); \textit{VLiteral } (\textit{Integer } 6)]
\]

	Furthermore, our goal can be proven with the derivation tree presented below. In this tree the trivial parts of the proofs are not described for readability (these are e.g. that the 5 and 6 expression literals evaluate to their value counterparts, the length of the expression or variable lists are the same as the evaluated value lists, etc.).
	
	\begin{equation}
	\nonumber
	\fontsize{9}{12}
	\begin{prooftree}
	\hypo{\coqdocvar{eval}\ \textit{``plus'' } [6; 5] = t}
	\infer1[\ref{OS:call}]{\bos{\{X : 6, Y : 5\}}{\varnothing}{X + Y}{t}}	
	\infer1[\ref{OS:let}]{\bos{\{X : 6\}}{\varnothing}{let\ Y = 5\ in\ X + Y}{t}}	
	\infer1[\ref{OS:let}]{\bos{\varnothing}{\varnothing}{let\ X = 6\ in\ let\ Y = 5\ in\ X + Y}{t}}
	\end{prooftree}
	\end{equation}
	
	The only remaining goal is to prove that $\coqdocvar{eval}\ \textit{``plus'' }\ [6; 5] = t$. We have already stated, that $t = \coqdocvar{eval}\ \textit{``plus'' } [5; 6]$, so it is sufficient to prove:
	
\[
	\coqdocvar{eval } \textit{``plus'' } [6; 5] = \coqdocvar{eval } \textit{``plus'' } [5; 6]
\]
	
	The commutativity can be used here (Theorem \ref{Prf:comm}), so we can swap the 5 and 6 values in the parameter list. After this modification, we get reflexivity.\qed
\end{proof}

With the same chain of thought, a more abstract refactoring also can be proved correct in our system.

\begin{example}[Swapping variable expressions]\label{ex:swap_abs}
If  we make the following assumptions:
\begin{align}
\nonumber
	&\bos{\Gamma}{\Delta}{e_1}{v_1} 
	&\bos{\Gamma + \{X : v_2\}}{\Delta}{e_1}{v_1} \\
\nonumber
	&\bos{\Gamma}{\Delta}{e_2}{v_2}
	&\bos{\Gamma + \{X : v_1\}}{\Delta}{e_2}{v_2}
\end{align}
then
	
\lstset{columns=fullflexible}

\begin{lstlisting}[language=Haskell]
  let X = e1 in let Y = e2 in X + Y
\end{lstlisting}
is equivalent to
\begin{lstlisting}[language=Haskell]
  let X = e2 in let Y = e1 in X + Y
\end{lstlisting}
	
\end{example}
\begin{proof}
	In a similar way to the Example \ref{ex:swap}, we reason like this.

\begin{align*}
&\forall \Gamma : Environment, \Delta : Closures, t : Value,\\
\nonumber
&\qquad\bos{\Gamma}{\Delta}{e_1}{v_1} \Longrightarrow
\bos{\Gamma + \{X : v_2\}}{\Delta}{e_1}{v_1}\Longrightarrow \\
\nonumber
&\qquad\bos{\Gamma}{\Delta}{e_2}{v_2}\Longrightarrow
\bos{\Gamma + \{X : v_1\}}{\Delta}{e_2}{v_2} \Longrightarrow\\
&
\bos{\Gamma}{\Delta}{\coqdocvar{ELet } [\Xvar]\ [\coqdocvar{$e_1$}]\ (\coqdocvar{ELet } [\Yvar]\ [\coqdocvar{$e_2$}]
\\&\qquad(\coqdocvar{ECall } \textit{``plus'' }\ [\coqdocvar{EVar } \Xvar ; \coqdocvar{EVar } \Yvar]))}{\coqdocvar{t}}
\Longleftrightarrow
\\&
\bos{\Gamma}{\Delta}{\coqdocvar{ELet } [\Xvar]\ [\coqdocvar{$e_2$}]\ (\coqdocvar{ELet } [\Yvar]\ [\coqdocvar{$e_1$}]
\\&\qquad(\coqdocvar{ECall } \textit{``plus'' }\ [\coqdocvar{EVar } \Xvar ; \coqdocvar{EVar } \Yvar]))}{\coqdocvar{t}}
\end{align*}

	The two directions of this equivalence are proved in exactly the same way, so only the forward ($\Longrightarrow$) direction is presented here.
	
	Now the main hypothesis has two \verb|let| statements in itself. Similarly to the Example \ref{ex:swap}, these statements could only be evaluated with rule \ref{OS:let}, i.e. there are two values ($v_1$ and $v_2$ because of the determinism and the assumptions) to which $e_1$ and $e_2$ evaluates:
	
	$\bos{\Gamma}{\Delta}{e_1}{v_1}$ and $\bos{\Gamma + \{X : v_1\}}{\Delta}{e_2}{v_2}$
	
	Moreover there appeared also a hypothesis: $\bos{\Gamma + \{X : v_1, Y : v_2\}}{\Delta}{X + Y}{\coqdocvar{t}}$. This hypothesis implies that $\coqdocvar{t} = \coqdocvar{eval } \textit{``plus'' } [v_1,v_2]$ because of the evaluation with \ref{OS:call}.
	
	Furthermore, the goal can be solved with the construction of a derivation tree. We denote $\Gamma + \{X : v_2, Y : v_1\}$ with $\Gamma_v$.
	
	\begin{equation}
	\nonumber
	\fontsize{9}{12}
	\begin{prooftree}
	\infer0{\bos{\Gamma}{\Delta}{e_2}{v_2}}
	\infer0[\ref{OS:call}]{\bos{\Gamma_v}{\Delta}{X + Y}{t}}
	\infer0{\bos{\Gamma + \{X : v_2\}}{\Delta}{e_1}{v_1}}
	\infer2[\ref{OS:let}]{\bos{\Gamma + \{X : v_2\}}{\Delta}{let\ Y = e_1\ in\ X + Y}{t}}
	\infer2[\ref{OS:let}]{\bos{\Gamma}{\Delta}{let\ X = e_2\ in\ let\ Y = e_1\ in\ X + Y}{t}}
	\end{prooftree}
	\end{equation}
	
	Now for the \coqdocvar{ECall}, the following derivation tree can be used.
	
	\begin{equation}
	\nonumber
	\fontsize{9}{12}
	\begin{prooftree}
	\hypo{\coqdocvar{get\_value}\ \Gamma_v\ Y = v_1}
	\infer1[\ref{OS:var}]{\bos{\Gamma_v}{\Delta}{Y}{v_1}}
	\hypo{\coqdocvar{get\_value}\ \Gamma_v\ X = v_2}
	\infer1[\ref{OS:var}]{\bos{\Gamma_v}{\Delta}{X}{v_2}}
	\hypo{\coqdocvar{eval}\ \textit{``plus'' } [v_2,v_1] = t}
	\infer3[\ref{OS:call}]{\bos{\Gamma_v}{\Delta}{X + Y}{t}}
	\end{prooftree}
	\end{equation}
	
	As mentioned before, $e_1$ and $e_2$ evaluates to $v_1$ and $v_2$ in the initial environment $\Gamma$ and also in the extended environments (for $e_1$ : $\Gamma + \{X : v_2\}$, for $e_2$ : $\Gamma + \{X : v_1\}$) too. So when the rule \ref{OS:let} applies, we can give a proof that $e_2$ and $e_1$ evaluates to $v_2$ and $v_1$.
	
	After making this statement, we can use the rule \ref{OS:call} to evaluate the \textit{``plus''}. The parameter variables will evaluate to $v_2$ and $v_1$. With this knowledge, we get: $\coqdocvar{eval}\ \textit{``plus'' } [v_2,v_1] = t$. As mentioned before $\coqdocvar{t} = \coqdocvar{eval } \textit{``plus'' } [v_1,v_2]$. So it is sufficient to prove, that:
	
\[
	\coqdocvar{eval } \textit{``plus'' } [v_2,v_1] = \coqdocvar{eval } \textit{``plus'' } [v_1,v_2]
\]
	
 The commutativity of \verb|eval| (Theorem \ref{Prf:comm}) can be used to solve this equality.\qed
\end{proof}

\begin{example}[Swapping variables in simultaneous let]\label{ex:swap_abs2}
	
\lstset{columns=fullflexible}
\begin{lstlisting}[language=Haskell]
  let <X, Y> = <e1, e2> in X + Y
\end{lstlisting}
is equivalent to
\begin{lstlisting}[language=Haskell]
  let <X, Y> = <e2, e1> in X + Y
\end{lstlisting}
	
\end{example}

\begin{proof}
	The proof for this example is very similar to the proof for Example \ref{ex:swap_abs}. The only difference is that one step is enough to evaluate the \verb|let| expressions, and that is the reason why no assumptions are needed.
	\qed
\end{proof}

\begin{example}[Function evaluation]\label{ex:fun_highlight}
	
\lstset{columns=fullflexible}
\begin{lstlisting}[language=Haskell]
  e
\end{lstlisting}
is equivalent to
\begin{lstlisting}[language=Haskell]
  let X = fun() -> e in
    apply X()
\end{lstlisting}
	
\end{example}

\begin{proof}
	In this case, both directions should be proved. At first, we formalise the problem:
	
	\begin{align*}
	&\forall \Gamma : Environment, \Delta : Closures, t : Value,
	\\
	&\qquad
	\bos{\Gamma}{\Delta}{e}{\coqdocvar{t}}
	\Longleftrightarrow
	\\&\qquad
	\bos{\Gamma}{\Delta}{\coqdocvar{ELet } [\Xvar]\ [\coqdocvar{EFun}\ []\ e]\ (\coqdocvar{EApply } (EVar \Xvar)\ []}{\coqdocvar{t}}
	&
	\end{align*}
	
	\noindent
	$\Longleftarrow$ direction:
	
	This can be proved by the construction of a derivation tree. We denote $\Gamma + \{X : \coqdocvar{VClosure}\ (\coqdocvar{inl}\ \Gamma)\ []\ e\}$ with $\Gamma_x$ and the value $\coqdocvar{VClosure}\ (\coqdocvar{inl}\ \Gamma)\ []\ e$ with $cl$ in the tree.
	
	\begin{equation}
	\nonumber
	\begin{prooftree}
	\infer0[\ref{OS:fun}]{\bos{\Gamma}{\Delta}{fun() \rightarrow e}{cl}}
			\infer0[\ref{OS:var}]{\bos{\Gamma_x}{\Delta}{X}{cl}}
			\hypo{\bos{\Gamma}{\Delta}{e}{t}}
		\infer2[\ref{OS:apply}]{\bos{\Gamma_x}{\Delta}{apply\ X()}{t}}
	\infer2[\ref{OS:let}]{\bos{\Gamma}{\Delta}{let\ X = fun() \rightarrow e\ in\ apply\ X()}{\coqdocvar{t}}}
	\end{prooftree}
	\end{equation}
	
	Only left to prove: $\bos{\Gamma}{\Delta}{e}{\coqdocvar{t}}$, but we have the same hypothesis.
	
	\noindent
	$\Longrightarrow$ direction:
	
	This can be proved by the deconstruction of the hypothesis for the \verb|let| expression. First only the \ref{OS:let} could be used for the evaluation. This means that the \coqdocvar{EFun} evaluates to some value, i.e. to the closure $\coqdocvar{VClosure}\ (\coqdocvar{inl}\ \Gamma)\ []\ e$. We get a new hypothesis: $\bos{\Gamma + \{X : \coqdocvar{VClosure}\ (\coqdocvar{inl}\ \Gamma)\ []\ e\}}{\Delta}{apply\ X()}{\coqdocvar{t}}$ (because rule \ref{OS:fun} and the determinism). Then the evaluation continued with the rule \ref{OS:apply}. This means, that the X variable evaluates to some closure (this one is in the environment, so only rule \ref{OS:var} could be applied) and the body of this closure evaluates to \coqdocvar{t} in the environment from the closure extended with the parameter-value bindings (in this case there is none). This means in our case: $\bos{\Gamma}{\Delta}{e}{\coqdocvar{t}}$ which is exactly what we want to prove.
	
	\qed
\end{proof}

To prove these examples in Coq, a significant number of lemmas were needed, such as the exposition of lists, the commutativity of the \verb|eval|, and so forth. However, the proofs mostly consist of the combination of hypotheses similar to the proofs in this paper. Although sometimes additional case separations were needed which resulted in lots of subgoals, these ones were solved very similarly, thus producing code duplication. In the future, these proofs should become simpler with the introduction of smart tactics and additional lemmas.

Moreover, in the concrete implementation for Example \ref{ex:swap_abs} we could use another formulation of the four additional assumptions: if \coqdocvar{$e_1$} does not contain the variables X and Y, then it will evaluate to the same value in the environments combined from these variables. This statement also stands for \coqdocvar{$e_2$}.

\subsection{Evaluation}\label{Sec:evaluation}

We showed that our formal semantics is a powerful tool. We managed to formalise and prove theorems, programs, program equivalence examples. This proves that the semantics is usable indeed. With this one we have a powerful tool to argue about sequential Core Erlang programs. In the previous sections we also mentioned some other approaches to formalise this semantics, and showed why our way is more usable for our purpose.

On the other hand, it also can be seen that this formalisation is not simple to use either in practice, partly because the Coq Proof Assistant makes its users write down everything (trivialities too). Of course this is a necessity of the correctness, however, this property results in complex proofs. As a possibility for future work, it would be very useful to create smart tactics, to simplify out proofs and examples. In addition, this semantics is not complete yet, so it cannot be used for any Core Erlang expression.

\section{Summary}\label{Sec:summary}

\subsection{Future work}

There are several ways to enhance our formalisation, we are going to focus mainly on these short term goals:
\begin{itemize}
	\item Extend semantics with additional expressions (e.g. \verb|try|);
	\item Handle errors (\verb|try| statement);
	\item Handle and log side effects;
	\item Create new lemmas, theorems and tactics to shorten the Coq implementation of the proofs;
	\item Formalise and prove more refactoring strategy.
\end{itemize}

\noindent Our long term goals include:
\begin{itemize}
	\item Advance to Erlang (semantics and syntax);
	\item Distinct primitive operations and inter-module calls;
	\item Formalize the parallel semantics too.
\end{itemize}

\noindent
The final goal of our project is to change the core of a scheme-based refactoring system to a formally verified core.

\subsection{Conclusion}

In this study, we discussed why a language formalisation is needed, then briefly the goal of our project (to prove refactoring correctness). To reach this objective, Erlang was chosen as the prototype language, then several existing Erlang formalisations were compared. Based on these ones, a new natural semantics was introduced for a subset of Erlang. This one was also formalised in Coq Proof Assistant along with essential theorems, proofs (like determinism) and formal expression evaluation examples. We also showed proofs about the meaning-preservation of simple refactoring strategies with our formal semantics. In the future, we intend to extend this formalisation with additional Erlang statements, error handling and more equivalence examples.

\section*{Acknowledgements}

The project has been supported by the European Union, co-financed by the European Social Fund (EFOP-3.6.2-16-2017-00013, ``Thematic Fundamental Research Collaborations Grounding Innovation in Informatics and Infocommunications (3IN)'').

Project no. ED\_18-1-2019-0030 (Application domain specific highly reliable IT solutions subprogramme) has been implemented with the support provided from the National Research, Development and Innovation Fund of Hungary, financed under the Thematic Excellence Programme funding scheme.

%
%
\bibliographystyle{splncs04}
\bibliography{bibliography}
\end{document}